\documentclass[runningheads]{llncs}
\usepackage[numbers,sort&compress]{natbib}

\usepackage{amsmath, amsthm}
\usepackage{graphicx, tikz, subcaption}
\usepackage{float}
\usepackage{enumitem}

\begin{document}
%
\title{On Tree Equilibria in Max-Distance Network Creation Games}
%
%
\author{Qian Wang}
%
\authorrunning{Q. Wang}
%
\institute{CFCS, Peking University, Beijing, China\\
\email{charlie@pku.edu.cn}}
\maketitle              
\begin{abstract}
We study the Nash equilibrium and the price of anarchy in the max-distance network creation game. The network creation game, first introduced and studied by \citet{fabrikant2003network}, is a classic model for real-world networks from a game-theoretic point of view. In a network creation game with $n$ selfish vertex agents, each vertex can build undirected edges incident to a subset of the other vertices. The goal of every agent is to minimize its creation cost plus its usage cost, where the creation cost is the unit edge cost $\alpha$ times the number of edges it builds, and the usage cost is the sum of distances to all other agents in the resulting network. The max-distance network creation game, introduced and studied by \citet{demaine2007price}, is a key variant of the original game, where the usage cost takes into account the maximum distance instead. The main result of this paper shows that for $\alpha > 19$ all equilibrium graphs in the max-distance network creation game must be trees, while the best bound in previous work is $\alpha > 129$~\cite{mihalak2010price}. We also improve the constant upper bound on the price of anarchy to $3$ for tree equilibria. Our work brings new insights into the structure of Nash equilibria and takes one step forward in settling the tree conjecture in the max-distance network creation game.

\end{abstract}

\section{Introduction}

Many important networks in the real world, e.g. the Internet or social networks, have been created in a decentralized, uncoordinated, spontaneous way. Modeling and analyzing such networks is an important challenge in the fields of computer science and social science. \citet{fabrikant2003network} studied the process of building these networks from a game-theoretic perspective by modeling it as a network creation game, which formalizes the cost function for each agent in the network as its creation cost plus its usage cost. Each agent, as a vertex, can choose to build an edge to any other agent at a prefixed constant cost $\alpha$. The creation cost is $\alpha$ times the number of edges it builds. The usage cost is the sum of the distances from this agent to all the others in the original game, while later variants of this game define the usage cost differently. Each agent aims to minimize its total cost by choosing a suitable subset of other agents to connect, which makes it a complex game on the graph. This sum-distance network creation game and its variants are increasingly applied in various areas and have recently been used to study payment channel networks on the blockchain~\cite{avarikioti2019payment,avarikioti2020ride,papadis2020blockchain}.

In the max-distance network creation game introduced by \citet{demaine2007price}, the usage cost of an agent is defined as the maximum distance from this agent to other agents. This variant covers a great portion of applications in practice. On the one hand, we can naturally think of it as a model with the assumption of risk aversion. Each agent always considers the worst case in which it needs to communicate with the farthest agent, while each agent considers the average case in the sum-distance network creation game. On the other hand, as interactions between agents can often be parallel in decentralized networks, such as social networks, the max-distance game could be a more realistic model than the sum-distance game if considering the non-additive time cost. For example, a post shared by one person on Facebook can be seen by multiple friends at the same time. Suppose that a person posts a nice photo and wants a group of people to see it as soon as possible. Some of them may not be his friends, but everyone who sees it is willing to share on his or her timeline. Then the person will care more about the maximum social distance rather than the sum of social distances to them.

No matter which cost function is chosen, the selfish agents in networks only make local decisions based on their own interests, so the equilibrium graph often fails to achieve global optimality, i.e., the lowest social cost. Thus, characterizing the structure of the equilibrium graph and calculating the price of anarchy (PoA)~\cite{koutsoupias1999worst} are the key parts of the research, where the price of anarchy is defined to be the ratio between the social cost of a worst Nash equilibrium and the optimal social cost~\cite{fabrikant2003network}. 

It has been proved that if all equilibrium graphs are trees, the price of anarchy will be constant in either sum-distance setting or max-distance setting~\cite{fabrikant2003network, mihalak2010price, bilo2018tree}. This result will help to analytically justify the famous small-world phenomenon~\cite{watts1999networks,kleinberg2000small}, which suggests selfishly built networks are indeed very efficient. Therefore, researchers pay special attention to the distinguishment between the tree structure and the non-tree structure in the network creation game. Intuitively, when the unit cost $\alpha$ for building an edge is high enough, a tree is likely to be an equilibrium graph due to the fact that each agent can already reach other agents and will have no incentive to build additional edges. 

For the sum-distance network creation game, \citet{fabrikant2003network} conjectured that every equilibrium graph is a tree for $\alpha > A$ where $A$ is a constant. Although this tree conjecture was disproved by \citet{albers2006nash}, the reformulated tree conjecture is well accepted as that in the sum-distance network creation game, every equilibrium graph is a tree for $\alpha > n$, where $n$ is the number of agents. A series of works have contributed to this threshold from $12n\log n$~\cite{albers2006nash} to $273n$~\cite{mihalak2010price}, $65n$~\cite{mamageishvili2015tree}, $17n$~\cite{alvarez2017network}, and to $4n-13$~\cite{bilo2018tree}. Recently, a preprint by \citet{dippel2021improved} claimed an improved bound of $3n-3$. 

For the max-distance network creation game, there is no work formally stating an adapted tree conjecture. However, by evaluating the difference between the usage costs (sum-distance and max-distance), we can roughly conjecture that every equilibrium graph is a tree for $\alpha > A$ where $A$ is a constant. \citet{mihalak2010price} proved $A$ is at most $129$ and gave the example that for $\alpha \leq 1$, a non-tree equilibrium graph could be constructed. 

The main contribution of this paper is to prove $A$ is at most $19$, i.e. for $\alpha > 19$, every equilibrium graph is a tree in the max-distance network creation game (Theorem~\ref{main}). We bring new insights into the structure of Nash equilibria and take one step forward in settling the tree conjecture in the max-distance network creation game. Our proof for the first time combines two orthogonal techniques, the classic degree approach~\cite{mihalak2010price, mamageishvili2015tree, alvarez2017network} and the recent min-cycle approach~\cite{alvarez2017network,bilo2018tree}. Moreover, we prove the upper bound on the price of anarchy is $3$ for tree equilibria (Theorem~\ref{poa}), so that we not only enlarge the range of $\alpha$ for which the price of anarchy takes a constant value, but also further tighten this constant upper bound from $4$ \cite{mihalak2010price} to $3$. The best known bounds on the price of anarchy are available in Table~\ref{improvement}. Our results can be interpreted as that the efficiency of decentralized networks is in fact better than previously known and demonstrated. 

\begin{table}
	\centering
	\caption{Summary of the known bounds on the PoA.}\label{improvement}
	\resizebox{\columnwidth}{!}{
		\begin{tabular}{c|c|c|c|c|c|c|c|c|c|c|c|c|c}
		\multicolumn{2}{c}{\hspace{0.25cm}$\alpha = 0$} & \multicolumn{2}{c}{\hspace{0.35cm}$\frac{1}{n-2}$} & \multicolumn{2}{c}{\hspace{0.55cm}$O(n^{-1/2})$} & \multicolumn{2}{c}{\hspace{1.45cm}$19$} & \multicolumn{2}{c}{\hspace{2.2cm}$n$} & \multicolumn{2}{c}{\hspace{0.95cm}$+\infty$} \\
	   
		& \multicolumn{2}{c|}{}  &  \multicolumn{2}{c|}{}  & \multicolumn{2}{|c|}{}  & \multicolumn{2}{|c|}{} & \multicolumn{2}{|c|}{}\\
		\cline{2-11}

		\text{PoA} & \multicolumn{2}{|c|}{1 (\cite{mihalak2010price})}  &  \multicolumn{2}{|c|}{$\Theta(1)$ (\cite{mihalak2010price})} & \multicolumn{2}{|c|}{$2^{O(\sqrt{\log n})}$  (\cite{mihalak2010price})}  & \multicolumn{2}{|c|}{$< 3$ (Corollary~\ref{main_corollary})} & \multicolumn{2}{|c|}{$\leq 2$ (\cite{demaine2007price})}\\
		\cline{2-11}

		\end{tabular}
	}
\end{table}

\section{Related Work}

\citet{fabrikant2003network} first introduced and studied the (sum-distance) network creation game. They proved a general upper bound of $O(\sqrt{\alpha})$ on the PoA. Then they found that numerous experiments and attempted constructions mostly yielded Nash equilibria that are trees with only a few exceptions, and postulated the tree conjecture that there exists a constant $A$, such that for $\alpha > A$, all Nash equilibria are trees. Moreover, if all Nash equilibria are trees, the price of anarchy is upper bounded by a constant $5$.

However, the tree conjecture was soon disproved by \citet{albers2006nash} by leveraging some interesting results on finite affine planes. They on the other hand showed that for $\alpha \geq 12n\lceil \log n\rceil $, every Nash equilibrium forms a tree and the price of anarchy is less than $1.5$. They also gave a general upper bound of $O(1+(\min \{ \alpha^2/n, n^2/\alpha \})^{1/3})$ on the price of anarchy for any $\alpha$. \citet{demaine2007price} further proved in the range $O(n^{1-\epsilon})$, the price of anarchy is constant and also improved the constant values in some ranges of $\alpha$. For $\alpha$ between $n^{1-\epsilon}$ and $12n\log n$, they gave an upper bound of $2^{O(\sqrt{\log n})}$.

By the work \citet{demaine2007price}, it is believed that for $\alpha \geq n$, every Nash equilibrium is a tree, as the reformulated tree conjecture. \citet{mihalak2010price} used a technique based on the average degree of biconnected components and significantly improved this range from $\alpha \geq 12n\lceil \log n\rceil $ to $\alpha > 273n$, which is asymptotically tight. Based on the same idea, \citet{mamageishvili2015tree} improved this range to $\alpha > 65n$ and \citet{alvarez2017network} further improved it to $\alpha > 17n$. Recently, \citet{bilo2018tree} improved the bound to $\alpha > 4n-13$, using a technique based on critical pairs and min cycles, which is orthogonal to the known degree approaches. \citet{dippel2021improved} proved an improved bound of $3n-3$ through a more detailed case-by-case discussion. By the above results, we immediately have that the price of anarchy is constant for $\alpha = O(n^{1-\epsilon})$ and $\alpha > 3n-3$. Moreover, \citet{alvarez2019price} proved that the price of anarchy is constant for $\alpha > n(1+\epsilon)$ by bounding the size of any biconnected component of any non-tree Nash equilibrium, but his newer result has no implication on the tree conjecture.

The max-distance network creation game, as a key variant of the original game, was introduced by \citet{demaine2007price}. They showed that the price of anarchy is at most $2$ for $\alpha \geq n$, $O(\min \{4^{\sqrt{\log n}} (n/\alpha)^{1/3} \})$ for $2\sqrt{\log n} \leq \alpha \leq n$, and $O(n^{2/\alpha})$ for $\alpha < 2\sqrt{\log n}$. \citet{mamageishvili2015tree} showed that the price of anarchy is constant for $\alpha > 129$ and $\alpha = O(n^{-1/2})$ and also proved that it is $2^{O(\sqrt{\log n})}$ for any $\alpha$. Specifically, they showed that in the max-distance network creation game, for $\alpha > 129$ all equilibrium graphs are trees. In this work, we combine the classic degree approach~\cite{mihalak2010price, mamageishvili2015tree, alvarez2017network} and the recent min-cycle approach~\cite{alvarez2017network,bilo2018tree} to show that for $\alpha > 19$ all equilibrium graphs are trees and the price of anarchy is constant. Compared to \citet{mihalak2010price}, we also improve the constant upper bound on the price of anarchy from $4$ to $3$ for tree equilibria.

It has been proved that computing the optimal strategy of a single agent in the network creation game with either cost function is NP-hard~\cite{fabrikant2003network,mihalak2010price}, making it impractical to verify whether a non-tree equilibrium graph exists through experimentation.

There are other variants of the network creation game, such as weighted network creation game proposed by \citet{albers2006nash}, bilateral network creation game studied by \citet{demaine2007price,corbo2005price}. For more studies on these variants, as well as exact or approximation algorithms for computing equilibria in network creation games, we refer readers to~\cite{demaine2007price,albers2006nash,corbo2005price,halevi2007network,baumann2008price,alon2013basic,ehsani2015bounded,lenzner2011dynamics,bilo2015max,kawald2013dynamics,chauhan2017selfish,cord2012basic,chauhan2016selfish,andelman2009strong}.

\section{Game Model}

In a network creation game, there are $n$ agents denoted by $\{1, 2, \ldots, n\}$. The strategy of agent $i$ is specified by a subset $s_i$ of $\{1, 2, \ldots, n\}\backslash \{i\}$, which corresponds to the set of neighbor agents to which agent $i$ creates links. Together, a strategy profile of this network is denoted by $s=\left(s_1, s_2, \ldots, s_n\right)$. The cost of creating a link is a constant $\alpha>0$ and the goal of every agent is to minimize its total cost. We only consider the pure Nash equilibrium in this work.

To define the cost function, we consider a graph $G_s=\left<V, E_s\right>$, where $V=\{v_1, v_2, \ldots, v_n\}$ represents $n$ agents and $E_s$ is determined by the strategy profile $s$, i.e. the links they create. $G_s$ has an edge $(u, v)\in E_s$ if either $u \in s_v$ or $v \in s_u$. Note all edges are undirected and unweighted. The distance between $u,v\in V$ is denoted by $d_{G_s}(u, v)$, the length of a shortest $u$-$v$-path in $G_s$. If $u$ and $v$ are not in the same connected component of $G_s$, $d_{G_s}(u, v)=+\infty$. The cost incurred by agent $v$ is defined as 
$$c_{G_s}(v) = \alpha |s_v| + D_{G_s}(v),$$ where $D_{G_s}(v) = \max_{u\in V} d_{G_s}(u, v)$. And the social cost is $$\text{Cost}(G_s) = \sum_{v\in V} c_{G_s}(v) = \alpha |E_s| + \sum_{v\in V}D_{G_s}(v).$$

A strategy profile $s$ is a Nash equilibrium, if for every agent $v$ and for every $s'$ which is different with $s$ only on $v$'s strategy, $c_{G_s}(v) \leq c_{G_{s'}}(v)$. In other words, no one has the incentive to deviate from its strategy as long as all other agents stick to their strategies. We call a graph $G_s$ an equilibrium graph if $s$ is a Nash equilibrium. Note that any equilibrium graph must be a connected graph, otherwise every agent will have an infinite cost. We also note that no edge is paid by both sides at equilibrium, otherwise either side can cancel its payment to get a lower cost.

We call a graph $G_{s^*}$ an optimal graph if $s^*$ minimizes the social cost such that $\text{Cost}(G_{s^*}) = \min_{s\in \mathcal{S} } \text{Cost}(G_s)$, where $S$ is the whole strategy space. The price of anarchy in this game is defined to be the ratio between the cost of a worst Nash equilibrium and the optimal social cost~\cite{fabrikant2003network}, i.e., $$\text{PoA}=\frac{\max_{s\in \mathcal{N} } \text{Cost}(G_s)}{\text{Cost}(G_{s^*})}, $$ where $\mathcal{N}$ is the set of all Nash equilibria.

In the following context, we will mostly omit the reference to the strategy profile $s$. If we need to represent two or more graphs induced by different strategy profiles, we will use $G, G'$, or $G_1, G_2, \ldots$ instead.

Next, we introduce some commonly used notations in graph theory. We use $V(G)$ and $E(G)$ to denote the vertex set and edge set of a graph $G$. We use $|\cdot|$ to denote the size of a vertex or edge set, e.g. $|V(G)|$. We also denote the length of a cycle $C$ by $|C|$. We denote the diameter of $G$ by $diam(G)=\max_{v\in V} D_G(v)$ and denote the radius of $G$ by $rad(G)=\min_{v\in V} D_G(v)$. A vertex is called a central vertex if $D_G(v) = rad(G)$. A cut vertex $v$ of $G$ is a vertex whose removal from $G$ increases the number of connected components. A biconnected graph is a graph with no cut vertex. A biconnected component $H$ of $G$ is a maximal biconnected subgraph with more than two vertices. Different from the standard definition of the biconnected component, we rule out the trivial subgraph with only two vertices and one edge for convenience. In this way, if there is no biconnected component in $G$, $G$ must be a tree.  The average degree of a biconnected component $H$ is $deg(H) = \frac{1}{|V(H)|}\sum_{v\in V(H)} deg_H(v)$, where $deg_H(v)$ counts the number of edges connecting $v$ in $H$. 

\section{Equilibrium Graphs}
\label{method}

This section demonstrates the main results of this paper. We show that in the max-distance network creation game, for $\alpha > 19$ all equilibrium graphs are trees and the price of anarchy is less than $3$. The main idea is to find the contradiction about the average degree of biconnected components. For any biconnected component $H$, we have $deg(H) \geq 2 + \frac{1}{5}$ when $\alpha > 5$ (Subsection~\ref{lower_section}), but meanwhile, we also have $deg(H) < 2 + 2/\lceil (\alpha - 1)/2 \rceil$ when $\alpha > 2$ (Subsection~\ref{upper_section}). Hence if $\alpha > 19$, a contradiction implies there is no biconnected component in $G$, thus a tree (Subsection~\ref{poa_section}). Finally, we give a upper bound of $3$ on the price of anarchy if all Nash equilibria are trees.

\subsection{Cycles in the Equilibrium Graph}
\label{cycles}

In this subsection, we give a basic characterization of cycles in equilibrium graphs. We show that in every equilibrium graph, all min cycles, if any, are directed (Lemma~\ref{directed}) and all cycles are at least of length $2\alpha - 1$ (Lemma~\ref{length2}). To prove these, we first prepare Lemma~\ref{crucial}, which is the most important lemma of this work, running through the whole paper.

\begin{lemma}
	\label{crucial}
	Let $G$ be an equilibrium graph and $a, b$ are two distinct vertices of $G$ such that $a$ buys an edge $(a, a_1)$. If there exists a shortest path tree $T$ of $G$ rooted at $b$ such that either $(a, a_1)$ is not an edge of $T$ or $a_1$ is the parent of $a$ in $T$, then $D_G(a) \leq D_G(b) + 1$. Furthermore, if $a$ is buying also another edge $(a, a_2), a_2\neq a_1$, and $(a, a_2)$ is not an edge of $T$, then $D_G(a) \leq D_G(b) + 1 - \alpha$.
\end{lemma}

\begin{proof}
	Consider the strategy change in which vertex $a$ swaps the edge $(a, a_1)$ with the edge $(a, b)$ and removes any other edge it owns but not contained in $T$, if any. Let $G'$ be the graph obtained after the swap. Observe that no vertex gets further to $b$, $d_{G'}(b, x) \leq d_{G}(b, x)$ for every $x\in V$, so $D_{G'}(a) \leq D_{G'}(b) + 1 \leq D_G(b) + 1$. Since $a$ does not apply this strategy change, we have $D_G(a) \leq D_G(b) + 1$. Furthermore, if $a$ also saves its creation cost $\alpha$ by removing $(a, a_2)$, we have $D_G(a) \leq D_{G'}(a) - \alpha \leq D_G(b) + 1 - \alpha$. See Figure~\ref{crucial_diagram}.
\end{proof}

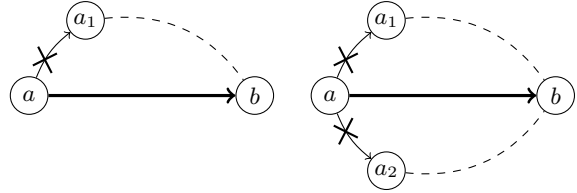
\begin{figure}[htpb]
    \centering
    \renewcommand*\thesubfigure{\arabic{subfigure}}
    \begin{tikzpicture}
        \begin{scope}
            \node[circle, draw, inner sep=0cm, minimum size=.5cm] (a) at (0, 0) {$a$};
            \node[circle, draw, inner sep=0cm, minimum size=.5cm] (b) at (3, 0) {$b$};
            \node[circle, draw, inner sep=0cm, minimum size=.5cm] (a1) at (0.75, 1) {$a_1$};
            \node[circle, draw, inner sep=0cm, minimum size=.5cm, opacity=0] (a2) at (0.75, -1) {$a_2$};
            \draw[->] (a) to [bend left=15] (a1);
            \draw[thick] (0.3, 0.3)--+(-0.18, 0.35);
            \draw[thick] (0.05, 0.4)--+(0.32, 0.16);
            \draw[->, very thick] (a) to  (b);
            \draw[dashed] (a1) to [bend left=30] (b);
        \end{scope}
        \begin{scope}[xshift=4cm]
            \node[circle, draw, inner sep=0cm, minimum size=.5cm] (a) at (0, 0) {$a$};
            \node[circle, draw, inner sep=0cm, minimum size=.5cm] (b) at (3, 0) {$b$};
            \node[circle, draw, inner sep=0cm, minimum size=.5cm] (a1) at (0.75, 1) {$a_1$};
            \node[circle, draw, inner sep=0cm, minimum size=.5cm] (a2) at (0.75, -1) {$a_2$};
            \draw[->] (a) to [bend left=15] (a1);
            \draw[->] (a) to [bend right=15] (a2);
            \draw[thick] (0.3, 0.3)--+(-0.18, 0.35);
            \draw[thick] (0.05, 0.4)--+(0.32, 0.16);
            \draw[thick] (0.3, -0.3)--+(-0.18, -0.35);
            \draw[thick] (0.05, -0.4)--+(0.32, -0.16);
            \draw[->, very thick] (a) to  (b);
            \draw[dashed] (a1) to [bend left=30] (b);
            \draw[dashed] (a2) to [bend right=30] (b);
        \end{scope}
    \end{tikzpicture}
    \caption{On the shortest path tree $T$ rooted at $b$, $a_1$ is the parent of $a$ or $(a, a_1)$ is a non-tree edge. The same goes for $a_2$. If $a$ swaps $(a, a_1)$ with $(a, b)$, the cost increases from $D_G(a)$ to at most $D_G(b)+1$. If $a$ further removes $(a, a_2)$, the cost increases from $D_G(a)$ to at most $D_G(b)+1 - \alpha$.}
    \label{crucial_diagram}
\end{figure}

\begin{lemma}[\cite{demaine2007price}]
	\label{length1}
	Every equilibrium graph has no cycle of length less than $\alpha + 2$.
\end{lemma}

\begin{proof}
	Suppose for contradiction there is a cycle $C$ in the equilibrium graph with length less than $\alpha + 2$ . Let $(v, u)$ be an edge of $C$. By symmetry, we can assume $v$ buys it. If $v$ removes this edge it decreases its creation cost by $\alpha$ and increases its usage cost by at most $|C|-2 < \alpha$. So $v$ has the incentive to remove this edge and this cannot be an equilibrium graph, a contradiction.
\end{proof}

\begin{definition}[Min Cycle]
    Let $G$ be a non-tree graph and let $C$ be a cycle of $G$. Then $C$ is a min cycle if for every two vertices $x_1, x_2 \in V(C)$, $d_C(x_1, x_2) = d_G(x_1, x_2)$.
\end{definition}

\begin{definition}[Directed Cycle]
	Let $G$ be a non-tree graph and let $C$ be a cycle of length $k$ in $G$. Then $C$ is a directed cycle if there is an ordering $u_0, \ldots, u_{k-1}$ of its $k$ vertices such that, for every $i=0, \ldots, k-1$, $(u_i, u_{(i+1)\mod k})$ is an edge of $C$ which is bought by vertex $u_i$.
\end{definition}

\begin{definition}[Antipodal Vertex]
	Let $C$ be a cycle of length $k$ and let $u, v$ be two vertices of $C$. Then $u$ is an antipodal vertex of $v$ if $d_C(u, v) \geq \lfloor \frac{k}{2} \rfloor$.
\end{definition}

\begin{lemma}
	\label{directed}
	For $\alpha > 2$, if $G$ is a non-tree equilibrium graph, then every min cycle in an equilibrium graph is directed.
\end{lemma}

\begin{proof}
	Suppose for contradiction, there is a min cycle $C$ in $G$ that is not directed. That means $C$ contains a vertex $v$ which buys both its incident edges in $C$, say $(v, v_1)$ and $(v, v_2)$.

	If $C$ is an odd-length cycle, then $v$ has two distinct antipodal vertices $u, u_1 \in C$. By symmetry we can assume $u$ buys the edge $(u, u_1)$. By Lemma~\ref{length1}, we have $|C| \geq \alpha + 2 > 4 \Longrightarrow |C| \geq 5$. Hence $v, v_1, v_2, u, u_1$ are different vertices. By Lemma~\ref{crucial} (where $a=u, a_1=u_1, b=v$), we have $D_G(u) \leq D_G(v) + 1$. Also by Lemma~\ref{crucial} (where $a=v, a_1=v_1, a_2=v_2, b=u$), we have $D_G(v) \leq D_G(u) + 1 - \alpha$. By summing up both the left-hand and the right-hand side of the two inequalities, we have $\alpha \leq 2$ which is a contradiction.
	
	If $C$ is an even-length cycle, then $v$ has one antipodal vertex $u\in C$. Denote the vertex adjacent to $u$ in $C$ (from any side) by $u_1$. If $u$ buys the edge $(u, u_1)$, the same analysis as above will lead to $\alpha \leq 2$, a contradiction. If $u_1$ buys the edge $(u_1, u)$, we can check $v$ and $u_1$ also satisfies the condition in Lemma~\ref{crucial} (where $a=v, b=u_1$). Therefore we have $D_G(u_1) \leq D_G(v) + 1$ and $D_G(v) \leq D_G(u_1) + 1 - \alpha$, which also leads to $\alpha \leq 2$, a contradiction.
\end{proof}

\begin{lemma}
	\label{contain}
	If $G$ is a non-tree equilibrium graph and $H$ is a biconnected component of $G$, then for every edge $e$ of $H$, there is a min cycle $C$ of $H$ that contains the edge $e$.
\end{lemma}

\begin{proof}
	Since $H$ is biconnected, there exists at least a cycle containing the edge $e$. Among all cycles in $H$ that contain the edge $e$, let $C$ be a cycle of minimum length. We claim that $C$ is a min cycle. For the sake of contradiction, assume that $C$ is not a min cycle. This implies that there are two vertices $u, v \in V_C$ such that $d_H(u, v) < d_C(u, v)$ and the shortest path $P'$ between $u$ and $v$ in $H$ is disjoint with $C$. Let $P_1$ and $P_2$ be two disjoint paths between $u$ and $v$ in $C$, with length $l_1$ and $l_2$, $l_1 + l_2 = |C|$. By symmetry we can assume $P_1$ contains $e$. Then $P$ and $P_1$ form a new cycle $C'$ with length $l' + l_1 = d_H(u, v) + l_1 < d_C(u, v) + l_1 \leq l_1 + l_2 = |C|$. Therefore we can find a cycle containing $e$ but strictly shorter than $C$, a contradiction.
\end{proof}

\begin{corollary}
	\label{buyone}
	For $\alpha > 2$, if $G$ is a non-tree equilibrium graph and $H$ is a biconnected component of $G$, then for every vertex $v \in V(H)$, $v$ buys at least one edge in $H$.
\end{corollary}

\begin{proof}
	Consider any edge of vertex $v$, say $(v, u)$. By Lemma~\ref{contain}, $(v, u)$ is contained in a min cycle $C$. Then by Lemma~\ref{directed}, we know $v$ buys at least one edge in $C$ if $\alpha > 2$.
\end{proof}

\begin{lemma}
	\label{length2}
	Every equilibrium graph has no cycle of length less than $2\alpha - 1$.
\end{lemma}

\begin{proof}
	For $\alpha \leq 2$, by Lemma~\ref{length1}, the length of any cycle is at least $\alpha + 2 > 2\alpha - 1$, so we assume $\alpha > 2$ from now on. 
	
	Since for every cycle in the graph, we can find a min cycle shorter than it, it is sufficient to prove that every equilibrium graph has no min cycle of length less than $2\alpha - 1$. Suppose for contradiction there is a min cycle $C$ in the equilibrium graph, whose length is less than $2\alpha - 1$. By Lemma~\ref{directed}, $C$ is a directed cycle. Let $(v, v_1)$ be an edge of $C$ bought by $v$. 
	
	If $C$ is an even-length cycle, we denote the antipodal vertex by $u$ and $u$ buys the edge $(u, u_1)$. Observe that the antipodal vertex of $u_1$ is $v_1$. If $C$ is an odd-length cycle, we denote the antipodal vertices by $u, u_1$ and w.l.o.g., assume that $u$ buys the edge $(u, u_1)$. Observe that the antipodal vertices of $u_1$ are $v$ and $v_1$. Thus, in both cases, we can find a shortest path tree $T$ rooted at $u_1$ such that $(v, v_1)$ is not contained in $T$, and $d_G(v, u_1) = \lfloor \frac{|C|-1}{2} \rfloor$. In this directed cycle $C$, $u_1$ also buys the edge $(u_1, u_2)$ and $u_2$ is the parent of $u_1$ in the shortest path tree rooted at $v$. By Lemma~\ref{crucial} (where $a = u_1, a_1 = u_2, b = v$), we have $D_G(u_1) \leq D_G(v) + 1$.

	Consider the strategy change in which $v$ removes the edge $(v, v_1)$ and $G'$ is the graph after the removal. For every vertex $x$, $d_{G'}(v, x) \leq d_{G'}(v, u_1) + d_{G'}(u_1, x) \leq \frac{|C|-1}{2} + d_G(u_1, x) < \alpha - 1 + D_G(u_1) \leq D_G(v) + \alpha$. Hence by applying this change, it decreases its creation cost by $\alpha$ and increases its usage cost by stricly less than $\alpha$, which is a contradiction.
\end{proof}

\subsection{Lower Bound for Average Degree}
\label{lower_section}

This subsection proves a lower bound of $2+\frac{1}{5}$ on the average degree of biconnected components (Lemma~\ref{lower}). The key idea is to show that any 2-degree path is no longer than $3$ (Lemma $\ref{path}$) so that the number of 2-degree vertices can be actually bounded.

\begin{lemma}
	\label{D_bound}
	For $\alpha > 2$, if $G$ is a non-tree equilibrium graph and $H$ is a biconnected component of $G$, then for every vertex $v \in V(H)$, $D_G(v) \leq rad(G) + 2$.
\end{lemma}

\begin{proof}
	Let $T$ be the shortest path tree rooted at some central vertex $v_0$ such that $D_G(v_0) = rad(G)$. Let $T_H$ be the intersection of $H$ and $T$, i.e., $$T_H = H\cap T = \left<V(H), E(H)\cap E(T)\right>.$$ Observe that $T_H$ is still a tree. By Corollary~\ref{buyone}, every vertex $v\in V(H)$ buys at least one edge in this biconnected component. Suppose $v$ buys the edge $(v, u)$ in $H$.
	
	If $(v, u)$ is a non-tree edge or $u$ is the parent of $v$, consider the strategy change in which $v$ swaps the edge $(v, u)$ to $(v, v_0)$. Since $v$ does not apply this change, we have $D_G(v) \leq D_{G'}(v) \leq D_{G'}(v_0) + 1 \leq D_G(v_0) + 1 = rad(G) + 1$. 

	If $v$ is the parent of $u$, we can always find such a vertex $s$ which is a descendant of $v$ and is also a leaf vertex of $T_H$. By Corollary~\ref{buyone}, $s$ buys at least one edge in $H$. Since it is already the leaf vertex of $T_H$, it could only buy a non-tree edge or a leading-up edge. Thus we have $D_G(s) \leq rad(G) + 1$. Then we consider the strategy change in which $v$ swaps the edge $(v, u)$ to $(v, s)$. Note in the shortest path tree rooted at $s$, $u$ is the parent of $v$. Since $v$ does not apply this change, we have $D_G(v) \leq D_{G'}(v) \leq D_{G'}(s) + 1 \leq D_G(s) + 1 \leq D_G(v_0) + 2 = rad(G) + 2$. 
\end{proof}

Let $H$ be an arbitrary biconnected component of the equilibrium graph $G$. We define a family of sets $\{S(v)\}_{v\in V(H)}$ such that
\[
S(v) = \{w\in V \mid v = \arg \min_{u \in V(H)} d_G(u, w)\}.
\]
Every vertex $w$ in $V$ is assigned to its closest vertex $v$ in $V(H)$. By the definition, we have 
\begin{itemize}
    \item $\cup_{v\in V(H)} S(v) = V$;
    \item for $u, v\in V(H)$, $S(u)\cap S(v) = \emptyset$ ;
    \item and for $v\in V(H)$, $S(v)\cap V(H) = \{v\}$. 
\end{itemize}


The following Lemma~\ref{outsider} shows for every vertex $v\in V(H)$, the furthest vertex never lies in $S(v)$, nor in $S(u)$ when $u$ is near to $v$.

\begin{lemma}
	\label{outsider}
	For $\alpha > 2$, if $G$ is a non-tree equilibrium graph and $H$ is a biconnected component of $G$, then for every vertex $v \in V(H)$ and every vertex $w \in S(v)$, $d_G(v, w) \leq D_G(v) + 2 - \alpha$.
\end{lemma}

\begin{proof}
	By Lemma~\ref{contain}, $v$ is contained in a min cycle $C$. By Lemma~\ref{length2}, $|C|\geq 2\alpha-1$. Hence there exists a vertex $u\in V(H)$ such that $d_G(v, u) \geq \frac{|C|-1}{2} \geq \alpha - 1$. By Lemma~\ref{crucial}, we know $D_G(u) \leq D_G(v) + 1$. As $H$ is a biconnected component, every shortest $u$-$w$-path contains $v$. Therefore, $d_G(v, w) = d_G(u, w) - d_G(v, u) \leq D_G(u) - (\alpha - 1) \leq D_G(v) + 2 - \alpha$.
\end{proof}

\begin{lemma}
	\label{path}
	For $\alpha > 5$, if $G$ is a non-tree equilibrium graph and $H$ is a biconnected component of $G$, then every path $x_0, x_1,\ldots, x_k, x_{k+1}$ in $H$ with $deg_H(x_i) = 2$ for $1 \leq i \leq k$, satisfies $k \leq 3$. Moreover, if $k=3$, we have $D_G(x_0)=rad(G)$ and $D_G(x_{k+1})\neq rad(G)$.
\end{lemma}

\begin{proof}
	Since $deg_H(x_i) = 2$ for $1 \leq i \leq k$, this path should be contained in some min cycle, which means it is a directed path by Lemma~\ref{directed}. W.l.o.g. we assume $x_i$ buys the edge $(x_i, x_{i+1})$ for $0 \leq i \leq k$.
	
	Suppose for contradiction there exists such a path with $k \geq 4$. By Lemma~\ref{D_bound}, we have $D_G(x_i) \in [rad(G), rad(G) + 2]$ for $0 \leq i \leq 3$. Therefore there exists $0 \leq j \leq 2$ such that $D_G(x_j) \geq D_G(x_{j+1})$. Consider the strategy change in which $x_j$ swaps the edge $(x_j, x_{j+1})$ to $(x_j, x_{j+3})$ and $G'$ is the graph after the swap. We will split the vertices $V$ into three parts (see Figure~\ref{four_set}) and show for vertex $w$ in any part, $d_{G'}(x_j, w) < D_G(x_j)$ for $\alpha > 5$, so that $D_{G'}(x_j) < D_G(x_j)$, a contradiction with equilibrium.

    \begin{figure}[h]
        \centering
        \begin{tikzpicture}
            \node[circle, draw, inner sep=0cm, minimum size=.5cm] (x0) at (-4, 0) {$x_0$};
            \node[circle, draw, inner sep=0cm, minimum size=.5cm] (x1) at (-3, 1.5) {$x_1$};
            \node[circle, draw, inner sep=0cm, minimum size=.5cm] (x2) at (-1, 2) {$x_2$};
            \node[circle, draw, inner sep=0cm, minimum size=.5cm] (x3) at (1, 2) {$x_3$};
            \node[circle, draw, inner sep=0cm, minimum size=.5cm] (x4) at (3, 1.5) {$x_4$};
            \node[circle, draw, inner sep=0cm, minimum size=.5cm] (x5) at (4, 0) {$x_5$};
            \draw[->] (x0) to (x1);
            \draw[->] (x1) to (x2);
            \draw[->, very thick] (x1) to [bend right=15] (x4);
            \draw[thick] (-1.8, 1.56)--+(-0.35, 0.35);
            \draw[thick] (-2.15, 1.56)--+(0.35, 0.35);
            \draw[->] (x2) to (x3);
            \draw[->] (x3) to (x4);
            \draw[->] (x4) to (x5);
            \draw (-3.9, 2.3) ellipse [x radius=1.8, y radius=0.8, rotate=140];
            \node[] at (-3.9, 2.3) {$S(x_1)$};
            \draw (-1.2, 3.2) ellipse [x radius=1.8, y radius=0.8, rotate=100];
            \node[] at (-1.2, 3.2) {$S(x_2)$};
            \draw (1.2, 3.2) ellipse [x radius=1.8, y radius=0.8, rotate=80];
            \node[] at (1.2, 3.2) {$S(x_3)$};
            \draw (-3, -0.7) ellipse [x radius=2.2, y radius=1.4, rotate=180];
            \node[] at (-3, -0.7) {$S_{left}$};
            \draw (2.6, 0) ellipse [x radius=2.5, y radius=1.6, rotate=40];
            \node[] at (2.6, 0) {$S_{right}$};
            \draw[dashed] (-1.2, -0.3) to (0.8, -0.3);
            \draw[dashed] (-1.2, -0.8) to (0.8, -0.8);
            \draw[dashed] (-1.2, -1.3) to (0.8, -1.3);
        \end{tikzpicture}
        \caption{Here is an example of the partition when $k=4$ and $j=1$. Consider $x_1$ swaps $(x_1, x_2)$ with $(x_1, x_4)$. First, by Lemma~\ref{outsider}, we know $x_1$ is very near to $S({x_1})$, $S({x_2})$ and $S({x_3})$ in $G$ while the swap only increases the distances a little. Then, by the definition of $S_{right}$, the swap makes $x_1$ strictly nearer to $S_{right}$. Finally, thanks to $D_G(x_1) \geq D_G(x_{2})$, $x_1$ can actually get strictly nearer to $S_{left}$ by taking a detour through $x_4$.}
        \label{four_set}
    \end{figure}
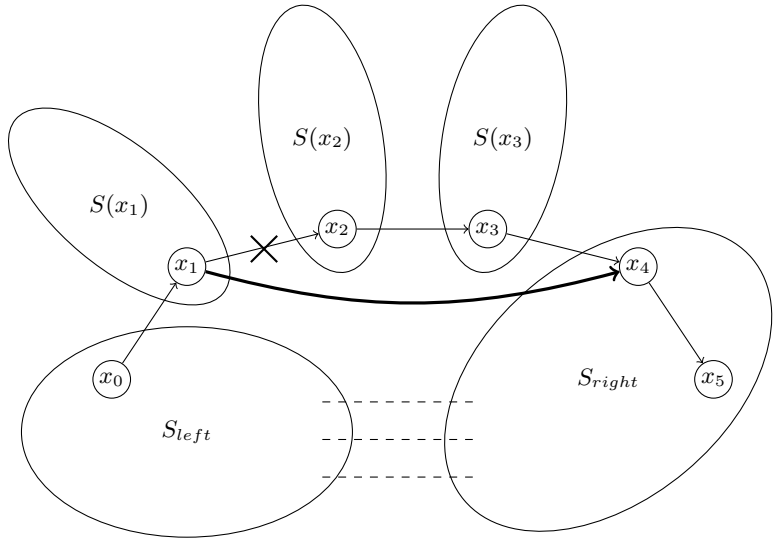

	\begin{enumerate}[label=(\alph*)]
		\item \textbf{Set $S(x_j)\cup S(x_{j+1})\cup S(x_{j+2})$:} By Lemma~\ref{outsider}, for every vertex $w\in S(x_j)$, $d_{G'}(x_j, w) = d_G(x_j, w) \leq D_G(v) + 2 - \alpha \leq rad(G) + 4 - \alpha < rad(G) \leq D_G(x_j)$. By Lemma~\ref{outsider}, for every vertex $w\in S(x_{j+t})$, $t=1,2$, $d_{G'}(x_j, w) \leq d_{G'}(x_j, x_{j+3}) + d_{G'}(x_{j+3}, x_{j+t}) + d_{G'}(x_{j+t}, w) = 1 + (3 - t) + d_{G}(x_{j+t}, w) \leq D_G(x_{j+t}) - t + 6 - \alpha < D_G(x_j)$.
		\item \textbf{Set $S_{right}$ where for every $w$ there exists a shortest $x_j$-$w$-path contains $x_{j+3}$}: For every vertex $w\in S_{right}$, $d_{G'}(x_j, w) \leq d_{G'}(x_j, x_{j+3}) + d_{G'}(x_{j+3}, w) = 1 + d_G(x_{j+3}, w) = d_G(x_j, w) - 2 < D_G(x_j)$.
		\item \textbf{Set $S_{left}$ where for every $w$ there is no shortest $x_j$-$w$-path contains $x_{j+3}$}: For every vertex $w\in S_{left}$ satisfying $d_G(x_j, w) < D_G(x_j)$, easy to see $d_{G'}(x_j, w) = d_G(x_j, w) < D_G(x_j)$. For every vertex $w\in S_{left}$ satisfying $d_G(x_j, w) = D_G(x_j)$, we claim there exists a shortest $x_{j+1}$-$w$-path contains $x_{j+3}$. Otherwise, $D_G(x_{j+1}) \geq d_G(x_{j+1}, w) = d_G(x_{j+1}, x_{j}) + d_G(x_{j}, w) = D_G(x_j) + 1 > D_G(x_j) \geq D_G(x_{j+1})$, a contradiction. Thus $d_{G'}(x_j, w) \leq d_{G'}(x_j, x_{j+3}) + d_{G'}(x_{j+3}, w) = 1 + d_G(x_{j+3}, w) = d_G(x_{j+1}, w) - 1 \leq D_G(x_{j+1}) - 1 < D_G(x_j)$.
	\end{enumerate}
	Combining them together, we have $D_{G'}(x_j) < D_G(x_j)$, contradicted with the equilibrium.
	
	When $k=3$, if $D_G(x_0) \geq D_G(x_1)$ or $D_G(x_1) \geq D_G(x_2)$, we can find a similar contradiction, so we must have $D_G(x_0) = rad(G), D_G(x_1)=rad(G)+1, D_G(x_2) = rad(G) + 2$. As $x_2$ does not choose to swap $(x_2, x_3)$ to $(x_2, x_4)$, we have $D_G(x_2) \leq D_G(x_4) + 1$, hence $D_G(x_{k+1}) = D_G(x_4) \geq rad(G) + 1$.
\end{proof}

\begin{corollary}
	\label{3neighbor}
	For $\alpha > 5$, if $G$ is a non-tree equilibrium graph and $H$ is a biconnected component of $G$, then for every vertex $v \in V(H)$, $v$ satisfies one of the following conditions:
	\begin{enumerate}[label=(\alph*)]
		\item there exists a vertex $u \in N_1(v)$ such that $deg_H(u) \geq 3$;
		\item $deg_H(u) = 2$ for $u \in N_1(v)$ and $deg_H(u) \geq 3$ for $u \in N_2(v)\backslash N_1(v)$, 
	\end{enumerate}
	where $N_k(v) = \{u\in H \mid d_H(u, v) \leq k\}$.
\end{corollary}

\begin{proof}
	Suppose for contradiction there exists a vertex $v_0$ satisfies none of the conditions. Since $v_0$ does not satisfy condition (a), we have a path $v_3, v_1, v_0, v_2, v_4$ such that $deg_H(v_i) = 2$ for $0\leq i \leq 2$. Since $v_0$ does not satisfy condition (b), we have $deg_H(v_3) = 2$ or $deg_H(v_4) = 2$. In either case, we get a $2$-degree path with length $k\geq 4$, contradicted with Lemma~\ref{path}.
\end{proof}

\begin{lemma}
	\label{lower}
	For $a > 5$, if $G$ is a non-tree equilibrium graph, then for every biconnected component $H$ of $G$, $deg(H) \geq 2 + \frac{1}{5}$.
\end{lemma}

\begin{proof}
	We assign every vertex $w\in H$ satisfying condition (a) in Corollary~\ref{3neighbor} to its closest vertex $v$ with $deg_H(v) \geq 3$, breaking ties arbitrarily. If $deg_H(w)\geq 3$, it is assigned to itself. Then we get a family of star-like vertex sets $\{V_i\}_{i=1}^m$, each of which is formed by a star-center vertex $v_i$ of degree at least $3$ and $|V_i|-1$ vertices of degree $2$ assigned to $v_i$. The total degree of them is $\sum_{i=1}^{m}deg_H(v_i) + 2 \sum_{i=1}^{m}(|V_i|-1)$. Note we have $|V_i|-1 \leq deg_H(v_i)$.

	Consider that the left vertices in $H$ form a vertex set $V_0$, every element of which satisfies condition (b) in Corollary~\ref{3neighbor}. The total degree of $V_0$ is $2|V_0|$. Each vertex in $V_0$ corresponds to a 2-degree path with $k=3$. Lemma~\ref{path} implies that any vertex cannot simultaneously be the start point of a 2-degree path with $k=3$ and the end point of another such path. Thus, every star-center vertex $v_i$ is linked to at most $(deg_H(v_i)-1)$ such paths because $v_i$ should be contained in some directed min cycle by Lemma~\ref{directed} and Lemma~\ref{contain}. Then we have $2|V_0|\leq \sum_{i=1}^{m}(deg_H(v_i)-1)$.

    The average degree of $H$ is
	\begin{align*}
		deg(H) &= \frac{\sum_{i=1}^{m}deg_H(v_i) + 2 \sum_{i=1}^{m}(|V_i|-1) + 2|V_0|}{\sum_{i=1}^{m}|V_i| + |V_0|} \\
		&= 2 + \frac{\sum_{i=1}^{m}deg_H(v_i)-2m}{\sum_{i=1}^{m}|V_i| + |V_0|} \\
		&\geq 2 + \frac{\sum_{i=1}^{m}deg_H(v_i)-2m}{\sum_{i=1}^{m}\left(deg_H(v_i) + 1\right) + \frac{1}{2}\sum_{i=1}^{m}
		(deg_H(v_i)-1)} \\
		&\geq 2 + \frac{3m - 2m}{\frac{3}{2}\cdot 3m + \frac{1}{2}m} \\
		&= 2 + \frac{1}{5}.
	\end{align*}
\end{proof}

\subsection{Upper Bound on Average Degree}
\label{upper_section}

This subsection proves that the average degree of any biconnected component is less than $2 + 2/\lceil (\alpha - 1)/2 \rceil$ (Lemma~\ref{upper}). We consider the shortest path tree $T$ rooted at some central vertex $v_0$. Let $T_H$ be the intersection of $H$ and $T$, i.e., $T_H = H\cap T = \left<V(H), E(H)\cap E(T)\right>$. Then the average degree of a biconnected component $H$ can be written as $$deg(H)=\frac{2(|E(T_H)| + |E(H) \backslash E(T_H)|)}{|V(T_H)|}.$$ Observe that $T_H$ is still a tree, so we have $|E(T_H)| = |V(T_H)| - 1$. Thus, to obtain the upper bound on average degree, we only need to bound the number of non-tree edges.

We call a vertex $u$ a shopping vertex if $u$ buys a non-tree edge. First observe that for $\alpha > 1$, every shopping vertex $u$ buys exactly one non-tree edge, otherwise by Lemma~\ref{crucial}, $D_G(u) \leq D_G(v_0) + 1 - \alpha < rad(G)$, a contradiction. The following Lemma~\ref{farshopping} aims to show that every pair of shopping vertices is not so close to each other on the tree, by which we can prove the number of shopping vertices has an upper bound.

\begin{lemma}
	\label{farshopping}
	For $\alpha > 2$, if $G$ is a non-tree equilibrium graph and $H$ is a biconnected component of $G$, then for every two shopping vertices $u_1, u_2\in V(H)$, $\max\left\{d_{T_H}(u_1, x), d_{T_H}(u_2, x)\right\} \geq \frac{\alpha - 1}{2}$, where $x$ is the lowest common ancestor of $u_1$ and $u_2$ in $T_H$.
\end{lemma}

\begin{proof}
	Let $(u_1, v_1)$ and $(u_2, v_2)$ be the non-tree edge bought by $u_1$ and $u_2$. Consider the path $P$ between $u_1$ and $u_2$ in $T_H$ and we extend it with $(u_1, v_1)$ and $(u_2, v_2)$. Then we get a path $x_0 = v_1, x_1 = u_1, x_2, \ldots, x_j, \ldots, x_{k-1}, x_k = u_2, x_{k+1} = v_2$, where $x_j$ is the lowest common ancestor of $u_1$ and $u_2$ in $T_H$, $j\in \{1, 2, \ldots, k\}$.

	Since $x_1$ buys $(x_1, x_0)$ and $x_k$ buys $(x_k, x_{k+1})$, there has to be a vertex $x_i$, $1\leq i\leq k$, such that $x_i$ buys both $(x_i, x_{i-1})$ and $(x_i, x_{i+1})$. Consider the strategy change in which $x_i$ removes both the edges $(x_i, x_{i-1})$ and $(x_i, x_{i+1})$ and buys the edge $(x_i, v_0)$ and let $G'$ be the graph after the strategy change. This change decreases the creation cost by $\alpha$ and increases the usage cost from $D_{G}(x_i)$ ($\geq D_{G}(v_0)$) to at most $D_{G'}(v_0) + 1$. In the rest of proof, we are going to show that if $\max\left\{d_{T_H}(u_1, x_j), d_{T_H}(u_2, x_j)\right\} < \frac{\alpha - 1}{2}$, then we have $D_{G'}(v_0) - D_{G}(v_0) + 1 < \alpha$, so that $x_i$ has the incentive to apply this strategy, which means $G$ cannot be an equilibrium graph.

	Suppose $d_{T_H}(u_1, x_j) < \frac{\alpha - 1}{2}$ and $d_{T_H}(u_2, x_j) < \frac{\alpha - 1}{2}$. First, we observe that $v_1$ and $v_2$ are not the descendants of any vertex $x_t$ for $t=1, 2, \ldots, k$. If $v_1$ is the descendant of $x_t$ for some $t$, $u_1, v_1$ and $x_t$ forms a cycle. Note that $v_1$ is at most one level deeper than $u_1$ in $T$, so the length of the cycle is at most $d_{T_H}(u_1, x_t) + d_{T_H}(v_1, x_t) + 1 \leq 2(d_{T_H}(u_1, x_j) + 1) < \alpha + 1$, contradicted with Lemma~\ref{length1} (or Lemma~\ref{length2}). The same goes for $v_2$.

	Second, we observe that only the vertices in $P$ and their descendants in $T$ may have increased distance to $v_0$ by this strategy change of $x_i$. By the previous analysis we know $v_0$-$v_1$-path and $v_0$-$v_2$ path are not affected, so $v_0$ can reach these vertices by taking a detour through $v_1$ or $v_2$. Let $y$ be one of the vertices in $P$ or a descendant of any vertex in $P$ and denote the nearest vertex to $y$ in $P$ by $x_{t^*}$, i.e. $x_{t^*} = \arg \min_{x_t} d_{T_H}(y, x_t)$. W.l.o.g., we assume $1\leq i\leq j$. 
	\begin{itemize}
		\item $1\leq i < j$: If $i < t^*\leq k$, since $v_0$-$x_j$-$x_{t^*}$-path is not affected, $d_{G'}(v_0, y) \leq d_{G}(v_0, y)$. If $t^* = i$, the edge $(x_i, v_0)$ makes the $v_0$-$x_i$-$y$-path even shorter. If $1 \leq t^* < i$, $v_0$ can reach $y$ by taking a detour through $v_1(= x_0)$ (see Figure~\ref{2remove}).
		
		\begin{equation}
			\label{detour}
			\begin{aligned}
				d_{G'}(v_0, y) \leq& d_{G'}(v_0, x_0) + d_{G'}(x_0, x_{t^*}) + d_{G'}(x_{t^*}, y) \\
				= & d_{G}(v_0, x_0) + d_{G}(x_0, x_{t^*}) + d_{G}(x_{t^*}, y) \\
				\leq & (d_{G}(v_0, u_1) + 1) + (1 + d_{G}(u_1, x_{t^*})) + d_{G}(v_0, y) - d_{G}(v_0, x_{t^*})\\
				= & d_{G}(v_0, y) + 2d_{G}(u_1, x_{t^*}) + 2 \\
				\leq & d_{G}(v_0, y) + 2d_{G}(u_1, x_j)\\
				< & d_{G}(v_0, y) + \alpha - 1.\\
			\end{aligned}
		\end{equation}
		\begin{figure}[h]
		    \centering
		    \begin{tikzpicture}
                \node[circle, draw, inner sep=0cm, minimum size=.8cm] (xi) at (0, 0) {$x_i$};
                \node[circle, draw, inner sep=0cm, minimum size=.8cm] (xi+1) at (0.6, 1.2) {$x_{i+1}$};
                \node[circle, draw, inner sep=0cm, minimum size=.8cm] (xi-1) at (-0.6, -1.2) {$x_{i-1}$};
                \node[circle, draw, inner sep=0cm, minimum size=.8cm] (xj) at (1.5, 2.4) {$x_j$};
                \node[circle, draw, inner sep=0cm, minimum size=.8cm] (xt) at (-1.8, -1.8) {$x_{t^*}$};
                \node[circle, draw, inner sep=0cm, minimum size=.8cm] (v0) at (1.5, 4.8) {$v_0$};
                \node[circle, draw, inner sep=0cm, minimum size=.8cm] (u1) at (-3, -2.4) {$u_1$};
                \node[circle, draw, inner sep=0cm, minimum size=.8cm] (v1) at (-4.5, -2.4) {$v_1$};
                \node[circle, draw, inner sep=0cm, minimum size=.8cm] (u2) at (2.5, 0) {$u_2$};
                \node[circle, draw, inner sep=0cm, minimum size=.8cm] (v2) at (4, 0) {$v_2$};
                \node[circle, draw, inner sep=0cm, minimum size=.8cm] (y) at (-0.6, -4) {$y$};
                \draw[dashed] (v0) to (xj);
                \draw[dashed] (xj) to (xi+1);
                \draw[dashed] (xj) to (u2);
                \draw[->] (u2) to (v2);
                \draw[<-] (xi+1) to (xi);
                \draw[->] (xi) to (xi-1);
                \draw[dashed] (xi-1) to (xt);
                \draw[dashed] (xt) to (y);
                \draw[dashed] (xt) to (u1);
                \draw[->] (u1) to (v1);
                \draw[dashed] (v1) to [bend left=30] (v0);
                \draw[dashed] (v2) to [bend right=30] (v0);
                \draw[->, very thick] (xi) to [bend left=30] (v0);
                \draw[thick] (-0.42, -0.42)--+(0.25, -0.35);
                \draw[thick] (-0.12, -0.48)--+(-0.35, -0.25);
                \draw[thick] (0.42, 0.42)--+(-0.25, 0.35);
                \draw[thick] (0.12, 0.48)--+(0.35, 0.25);
		    \end{tikzpicture}
		    \caption{Here we illustrate the case when $1\leq t^* < i < j$. In the new graph $G'$, $v_0$ can reach $y$ through $v_1$, $u_1$ and $x_{t^*}$. Note that $v_1$ is at most one level deeper than $u_1$. From $v_0 \to x_j \to x_{t^*} \to y$ to $v_0 \to v_1 \to x_{t^*} \to y$, the distance is increased by up to $2d_G(v_1, x_{t^*}) \leq 2d_G(u_1, x_j) < \alpha - 1$.}
			\label{2remove}
		\end{figure}
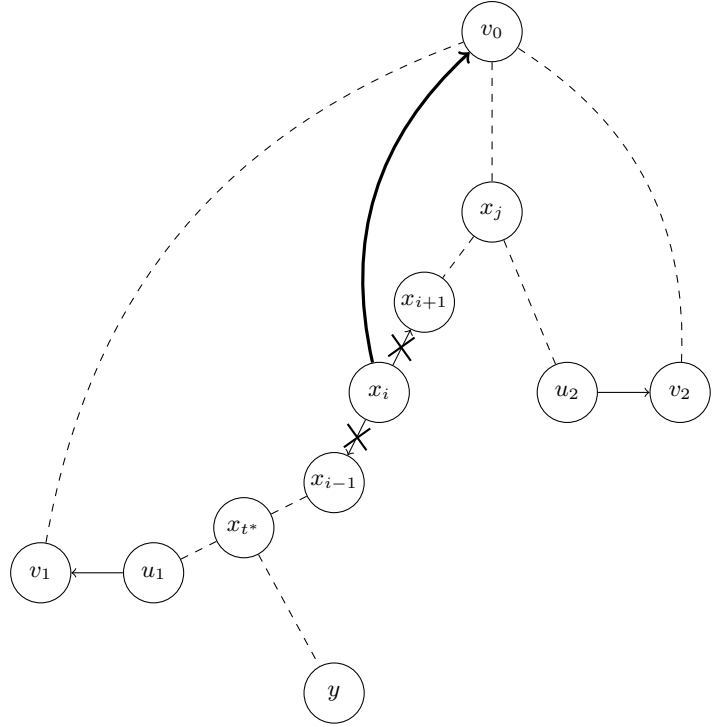
		\item $i = j$: If $t^* = i$, the edge $(x_i, v_0)$ makes the $v_0$-$x_i$-$y$-path even shorter. If $t^* \neq i$, $v_0$ can reach $y$ by taking a detour through $v_1$ or $v_2$. In either case we can get $d_{G'}(v_0, y) < d_{G}(v_0, y) + \alpha - 1$ by the inequality similar to Equation (\ref{detour}).
	\end{itemize}
	Combining them together we have $D_{G'}(v_0) < D_{G}(v_0) + \alpha - 1$, which finishes the proof.
\end{proof}

\begin{corollary}
	For $\alpha > 2$, if $G$ is a non-tree equilibrium graph and $H$ is a biconnected component of $G$, then for every two shopping vertices $u_1, u_2\in V(H)$, $d_{T_H}(u_1, u_2) \geq \frac{\alpha - 1}{2}$.
\end{corollary}

\begin{lemma}
	\label{upper}
	For $\alpha > 2$, if $G$ is a non-tree equilibrium graph, then for every biconnected component $H$ of $G$, $deg(H) < 2 + 2/\lceil (\alpha - 1)/2 \rceil$.
\end{lemma}

\begin{proof}
	For each shopping vertex $u$, we define such a vertex set $A(u)$ that for every vertex $w \in A(u)$, $w=u$ or $w$ is an ancestor of $u$ in $T_H$, and $d_{T_H}(u, w) < \frac{\alpha - 1}{2}$. $|A(u)| = \lceil \frac{\alpha - 1}{2} \rceil $ for every shopping vertex $u$.
	
	If for every pair of shopping vertex $u_1, u_2$, $A(u_1)\cap A(u_2) = \emptyset$, the number of shopping vertices could be bounded by $|V(H)|/\lceil (\alpha - 1)/2  \rceil$, then the desired result can be obtained by $$deg(H)=\frac{2(|E(T_H)| + |E(H) \backslash E(T_H)|)}{|V(T_H)|} < 2 + \frac{2}{\lceil (\alpha - 1)/2 \rceil}.$$

	Suppose for contradiction there are two shopping vertices $u_1, u_2$, such that $A(u_1)\cap A(u_2) \neq \emptyset$. Denote the lowest common ancestor of $u_1$ and $u_2$ by $x$. By the definition of $A$, we must have $x\in A(u_1)\cap A(u_2)$. However by Lemma~\ref{farshopping}, $\max\left\{d_{T_H}(u_1, x), d_{T_H}(u_2, x)\right\} \geq \frac{\alpha - 1}{2}$, which means $x$ cannot belong to $A(u_1)$ and $A(u_2)$ simultaneously. This contradiction completes the proof.
\end{proof}

\subsection{Tree Nash Equilibria and Price of Anarchy}
\label{poa_section}

Combining the lemmas from the above two subsections, we are ready to give our main result on tree Nash equilibria in Theorem~\ref{main}. Then we will give the upper bound on the price of anarchy for tree equilibria in Theorem~\ref{poa}. 

\begin{theorem}
	\label{main}
	For $\alpha > 19$, every equilibrium graph is a tree.
\end{theorem}

\begin{proof}
    $\alpha > 19 \Longrightarrow \left\lceil \frac{\alpha - 1}{2}\right\rceil \geq 10$. By Lemma~\ref{upper}, if $G$ is a non-tree equilibrium graph, for every biconnected component $H$ of $G$, we have $deg(H) < 2 + \frac{1}{5}$, contradicted with Lemma~\ref{lower}.
\end{proof}
    

\begin{lemma}[\cite{demaine2007price}]
	\label{optimum}
	For $\alpha \leq \frac{2}{n-2}$, the complete graph is a social optimum. For $\alpha \geq \frac{2}{n-2}$, the star is a social optimum.
\end{lemma}

\begin{lemma}[\cite{demaine2007price}]
	\label{smallalpha}
	For $\alpha < \frac{1}{n-2}$, the price of anarchy is 1. For $\alpha < \frac{2}{n-2}$, the price of anarchy is at most 2.
\end{lemma}

\begin{theorem}
	\label{poa}
	If all Nash equilibria are trees, the price of anarchy is less than $3$.
\end{theorem}

\begin{proof}
	The claim is trivial when $n \leq 2$, or when $\alpha < \frac{2}{n-2}$ by Lemma~\ref{smallalpha}. We therefore assume $n \geq 3$, and $\alpha \geq \frac{2}{n-2}$. Let $G$ be a tree equilibrium graph and $T$ be the shortest path tree rooted at some central vertex $v$ such that $D_G(v) = rad(G) = \lceil \frac{diam(G)}{2} \rceil$ (observe we can find such a central vertex in a tree graph). Let $u$ be a deepest leaf vertex of $T$ such that $d_G(u, v) = rad(G)$ and $D_G(u) = diam(G)$. Consider the strategy change in which $u$ additionally buys the edge $(u, v)$ and $G'$ is the graph after the change. Then $D_{G'}(u) \leq D_{G'}(v) + 1 \leq D_G(v) + 1$. The usage cost decreases by at least $D_{G}(u) - D_{G'}(u) \geq \lfloor \frac{diam(G)}{2} \rfloor -1$ and therefore $\alpha > \lfloor \frac{diam(G)}{2} \rfloor -1 \Longrightarrow diam(G) \leq 2\alpha + 3$.

	We define $N_k^=(v) = \{u\in G \mid d_G(v, u) = k\}$ and the total cost is 
	\begin{equation*}
		\begin{aligned}
			\text{Cost}(T) & = (n-1)\alpha + \sum_{u\in V} D_G(u) \\
			& \leq (n-1)\alpha + \sum_{i=0}^{rad(G)} |N_i^=(v)|(rad(G)+i)\\
			& = (n-1)\alpha + n\cdot rad(G) + \sum_{i=1}^{rad(G)} |N_i^=(v)|\cdot i.\\
		\end{aligned}
	\end{equation*}

	Observe that $|N_k^=(v)| \geq 2$ for $k=1, 2, \ldots, rad(G)-1$. To get the upper bound on $\text{Cost}(T)$, we can consider the worst case where $|N_k^=(v)| = 2$ for $k=1, 2, \ldots, rad(G)-1$ and $N_{rad(G)}^=(v) = V\backslash \bigcup_{k=0}^{rad(G)-1}N_k^=(v)$. By Lemma~\ref{optimum}, a star is a social optimum, so we get
 
	\begin{equation*}
		\begin{aligned}
			\text{PoA} & = \frac{\max_T \text{Cost}(T)}{\text{Cost}(G^*)} \\
			& \leq \frac{(n-1)\alpha + n\cdot rad(G) + rad(G)(rad(G)-1) + (n - 2rad(G) + 1)rad(G)}{(n-1)\alpha + 2n - 1} \\
			& = \frac{(n-1)\alpha + 2n\cdot rad(G) - (rad(G))^2}{(n-1)\alpha + 2n - 1} \\
			& = 1 + \frac{2(n-1)(rad(G) - 1) - (rad(G)-1)^2}{(n-1)\alpha + 2n - 1}\\
			& < 1 + \frac{2(rad(G)-1)}{\alpha +2} \\
			& < 1 + \frac{diam(G) - 1}{\alpha + 2} \\
			& < 3. \\
		\end{aligned}
	\end{equation*}
\end{proof}

\begin{corollary}
	\label{main_corollary}
	For $\alpha > 19$, the price of anarchy is less than $3$.
\end{corollary}

\section{Conclusion}

In this paper, we study the max-distance network creation game. We use the technique based on the average degree of the biconnected component and also combine the analysis of min cycles. We show that for $\alpha > 19$ there will be a contradiction if there exists a biconnected component in the equilibrium graph. Therefore, every equilibrium graph is a tree for $\alpha > 19$ in the max-distance network creation game and the price of anarchy in this range is a constant. Moreover, we improve the upper bound on the price of anarchy for tree equilibria to $3$.

It would be interesting to determine the optimal bound in future work. By saying the optimal bound, we refer to both the minimal unit edge cost $\alpha$ for every equilibrium graph to be a tree, and the tightest constant upper bound on the price of anarchy. It can be seen from our proof that most lemmas only require $\alpha > 2$ or $\alpha > 5$, which implies there is still a lot of room for improvement of these bounds. 
Another direction of future work is to study how to combine the degree approach and the min-cycle approach to obtain better bounds in the sum-distance network creation game. 



%
%
%
\bibliographystyle{splncs04nat}
\bibliography{ref}
%




\end{document}